\documentclass[11pt]{article}
\usepackage{fullpage}
\usepackage{amssymb}
\usepackage{amsmath,amsthm}
\usepackage{multirow}
\usepackage{color}
\usepackage{xspace}

\usepackage{ifpdf}
\ifpdf    
\usepackage{hyperref}
\else    
\usepackage[hypertex]{hyperref}
\fi

\newcommand{\remove}[1]{}

\newtheorem{definition}{Definition}[section]
\newtheorem{theorem}{Theorem}[section]

\newtheorem{lemma}[theorem]{Lemma}

\DeclareMathOperator{\ddim}{ddim}
\DeclareMathOperator{\argmin}{argmin}

\DeclareMathOperator{\diam}{diam}
\DeclareMathOperator{\roott}{root}
\DeclareMathOperator{\levell}{level}
\DeclareMathOperator{\wt}{\mathsf{wt}}
\DeclareMathOperator{\median}{med}
\DeclareMathOperator{\cntr}{cntr}
\def\topp{\mathrm{top}}
\def\summ{\mathsf{sum}}

\newcommand{\aset}[1]{\{#1\}}

\newcommand{\ceil}[1]{\lceil{#1}\rceil}

\def\iroot{\ensuremath{{i_{\mathrm{root}}}}}
\def\ifin{\ensuremath{{i^*}}}

\def\OPT{\ensuremath{{\mathrm{OPT}}}}
\def\ALG{\ensuremath{{\mathrm{ALG}}}}
\def\veps{\ensuremath{\varepsilon}}

\newcommand{\comment}[1]{}
\newcommand{\junk}[1]{}

\begin{document}
{{
\title{Faster Clustering via Preprocessing%
\thanks{This work was supported in part by The Israel Science Foundation
(grant \#452/08), by a US-Israel BSF grant \#2010418,
and by the Citi Foundation.}
}
\author{
Tsvi Kopelowitz \qquad Robert Krauthgamer
\\ Weizmann Institute of Science
\\ \texttt{\{tsvi.kopelowitz,robert.krauthgamer\}@weizmann.ac.il}
}


\maketitle

\thispagestyle{empty}
\setcounter{page}{0}

\begin{abstract}
We examine the efficiency of clustering a set of points,
when the encompassing metric space may be preprocessed in advance.
In computational problems of this genre,
there is a first stage of preprocessing, whose input is a collection of points $M$;
the next stage receives as input a query set $Q\subset M$,
and should report a clustering of $Q$ according to some objective,
such as $1$-median, in which case the answer is
a point $a\in M$ minimizing $\sum_{q\in Q} d_M(a,q)$.

We design fast algorithms that approximately solve such problems
under standard clustering objectives like $p$-center and $p$-median,
when the metric $M$ has low doubling dimension.
By leveraging the preprocessing stage, our algorithms achieve
query time that is near-linear in the query size $n=|Q|$,
and is (almost) independent of the total number of points $m=|M|$.
\end{abstract}

\newpage
\section{Introduction}
Clustering is a ubiquitous computational task of prime importance
in numerous applications and domains, including
machine learning, image processing, and bioinformatics.
While the clustering problem has several variations,
it often falls within the following framework of \emph{metric clustering}:
given a set of points $Q$ in a metric space $(M,d)$,
choose a set of centers $C$ (in that same metric space)
so as to minimize some objective function of $Q$ and the centers $C$.
For example, in the \emph{$p$-median} problem,
the goal is to find a set of $p$ centers $C\subseteq M$ that minimizes the objective
$$\textstyle
  \median(Q,C):= \sum_{q\in Q} d(q,C),
$$
where we define $d(q,C):=\min_{c\in C} d(q,c)$.

Our focus here is on understanding whether an initial preprocessing stage
can speed up the process of (metric) clustering.
Concretely, we are interested in algorithms for efficient clustering of $Q$
when the metric $M$ can be preprocessed in advance.
Throughout, we denote the number of \emph{center candidates} by $m=|M|$,
and the number of \emph{query points} by $n=|Q|$.
The goal is to answer queries with time close to linear in $n$
and (almost) independent of $m$.
To our knowledge, no previous research on (metric) clustering problems
has addressed the issue of preprocessing.
Past work has largely focused on the offline problem,
where the entire input is given at once,
either because $M$ is implicit (e.g., a Euclidean space)
or because $M$ is given together with $Q$ (called discrete centers).
Other past work studied the online version,
where points arrive one by one (the data-stream model).

Clustering with preprocessing can model, for example, the following scenario.
Consider a huge corpus of documents ($M$)
with distances between the documents ($d$) defining a metric space.
Given a relatively small subset of the documents ($Q$),
we may wish to quickly cluster them using centers from the corpus.
Since preprocessing needs to be done only once,
it has the benefit that even a huge corpus can be processed,
by pooling together many machines or by running it for several days.

The first problem we consider is the $p$-median problem defined above.
A second problem of interest, called \emph{$p$-center}, is to find a set of $p$ centers $C\subseteq M$ that minimizes the objective
$$ \textstyle
  \cntr(Q,C):= \max_{q\in Q}\{d(q,C)\}.
$$
Observe that when $n=1$ and $p=1$, both the $p$-median and $p$-center problems
receive a single input point $q$
and seek the point of $M$ that is closest to $q$,
which is precisely the famous nearest neighbor search (NNS) problem.
Even for this special case of NNS (i.e., $n=p=1$), Krauthgamer and Lee~\cite{KL04b} have shown
that achieving approximation factor better than $\frac{7}{5}$ in a general metric
requires query time that depends on the doubling dimension of the metric,
regardless of the preprocessing.
(Throughout, we denote the doubling dimension by $\ddim=\ddim(M)$;
see Section~\ref{sec:prelims} for a formal definition.)
It thus follows that for general $n$ and $p$, one must consider metrics $M$
whose doubling dimension is bounded, and we indeed assume as such.
We also assume that computing the distance between two points takes $O(1)$ time.
This whole approach follows an established line of research that covers
a host of problems
including nearest neighbor search \cite{KL04,HM06,BKL06,CG06}
as well as routing, distance estimation, the traveling salesman problem
and classification
see e.g. \cite{AGGM06,KRX08}, \cite{KSW09,Slivkins07},
\cite{Talwar04,BGK12}, and \cite{BLL09,GKK10}.

\subsection{Results}
\label{sec:results}

We provide the first clustering algorithms that leverage a preprocessing stage
to obtain improved query time.
Specifically, we design algorithms that compute $(1+\veps)$--approximation
for the $p$-median and $p$-center problems;
the precise time and space bounds are presented in Table~\ref{tab:results}.
Observe that the query time is near-linear in $n$
and is (almost) independent of $|M|$,
assuming the other parameters ($\veps^{-1}$, $p$ and $\ddim$) are small.
For sake of simplicity, we let our results depend on the aspect ratio of $M$,
denoted $\Delta=\Delta(M)$.
Such bounds can usually be refined,
replacing e.g.\ $\log\Delta$ terms with $\log n$,
by adapting our algorithms using known techniques and data structures,
but it would clutter the presentation of our main ideas.
Interestingly, we use essentially the same data structure for all problems solved.
All our space bounds are expressed in terms of machine words, which as usual
can accommodate a pointer to a data point or a single distance value.

\begin{table}[htb]
\begin{center}
\begin{tabular}{|l|l|l|l|}
\hline
Problem
  & \multicolumn{1}{l|}{\ Preprocessing time}
  & \multicolumn{1}{l|}{\ Space}
  & \multicolumn{1}{l|}{\ Query Time}
  \\
\hline
\hline
$1$-median
  & \multirow{2}{*}{ $2^{O(\ddim)} m\log \Delta \log\log \Delta$ }
  & \multirow{2}{*}{ $2^{O(\ddim)} m$ }
  & \multirow{2}{*}{ $O(n\log n + 2^{O(\ddim)}\log \Delta + \veps^{-O(\ddim)})$ }
\\
  \, Theorem \ref{thm:eff1med}
  & & & \\
\hline
$p$-median
  & \multirow{2}{*}{ $2^{O(\ddim)} m\log \Delta\log\log \Delta$ }
  & \multirow{2}{*}{ $2^{O(\ddim)} m$ }
  & { $O(n\log n + {\veps ^{-O(p\cdot \ddim)}} (p\log n)^{O(p)} $ }
\\
  \, Theorem \ref{thm:p_median}
  &
  &
  & \,$\mbox{} + \veps^{-O(\ddim)}(p\log n)^{O(1)} \cdot \log \log \log \Delta) $ \\
\hline
$p$-center
  & \multirow{2}{*}{ $2^{O(\ddim)} m\log \Delta\log\log \Delta$ }
  & \multirow{2}{*}{ $2^{O(\ddim)} m$ }
  & $O(n\log n + p\log \log \log \Delta $
\\
  \, Theorem \ref{thm:p_center}
  &
  &
  & \,$\mbox{} + p^{p+1}\cdot \veps^{-O(p\cdot \ddim)})$\\
\hline
\end{tabular}
\end{center}
  \caption{Our algorithms for $(1+\veps)$--approximation of clustering problems.}
  \label{tab:results}
\end{table}

We point two possible extensions of our results.
First, one may ask about updates to $M$, i.e., inserting and deleting points.
Our data structure is similar to previous work on NNS,
and thus we expect the methods known there (see e.g. \cite{KL04})
to apply also in our case, although we did not check all the details.
Second, we assume throughout that $Q\subset M$.
One may remove this restriction, possibly adapting the definition of
$\ddim$ and $\Delta$ to refer to $M\cup Q$.
Again, we have not checked the details, but we expect this is possible
by roughly applying the procedure of inserting $Q$ to $M$
before executing the query $Q$,
except that now we cannot use points of $Q\setminus M$ as centers.

Our bounds for clustering with preprocessing
are in a new model that was not studied before,
and thus cannot be compared directly with previous work.
But of course, all of our results immediately imply also
algorithms for the respective offline problems,
where the input includes both $Q$ and $M$.
Since our preprocessing time is near linear in $m$,
these are pretty efficient as well.
Even for the offline problems, our results are new,
as we are not aware of previous work on clustering ($p$-median and $p$-center)
in metric spaces of bounded doubling dimension.
Notice that a naive algorithm,
which exhaustively tries all possible sets of centers (with no preprocessing)
finds an optimal solution but takes runtime ${m \choose p}np$,
which is significantly higher even for $p=1$.
Another possible comparison is with the respective Euclidean problems;
this is only for the sake of analogy and is discussed in Section \ref{sec:related}.
We also point out that our 1-median algorithm is deterministic,
while previous algorithms achieving $(1+\veps)$--approximation for 1-median,
even in Euclidean metrics, are randomized \cite{Indyk99,BHI02,KSS10}.

\subsection{Techniques}
\label{sec:techniques}

Our algorithms build on several techniques from prior work.
The common algorithmic paradigm for the NNS problem
in metrics with low doubling dimension~\cite{KL04,HM06,BKL06,CG06}
(which in our context is just the special case $p=n=1$),
is to look for an answer (center point)
by restricting attention to a sequence of search balls,
whose radii are decreasing, usually by a constant factor.
When the ball's radius becomes small enough,
the algorithms revert to exhaustively trying
a small set of candidates inside the search ball,
with the property that at least one candidate in the set
must be a good enough approximation to an optimal answer.
Such a set of candidates is sometimes called a \emph{centroid set} \cite{HM04}.
We follow this paradigm, but extend and modify it for our needs.

We further borrow a technique of constructing a \emph{coreset} \cite{BHI02},
which essentially assigns points in $Q$ to a small set of ``representatives'' $R$,
so that solving the clustering problem on the weighted set $R$
provides a good approximation for clustering $Q$.
The weight of a representative $r\in R$ is simply
the number of query points $q\in Q$ assigned to it.
In contrast to previous work on coresets and on centroid sets,
we have the leverage of preprocessing $M$,
and our challenge is to quickly construct such sets for $Q$ during query time.

We also devise a new technique
(new at least in the context of our clustering problems)
of ``projecting'' the data structure
constructed for $M$ (during preprocessing) onto the query set $Q\subset M$.
While a data structure for $Q$ can be constructed from scratch
in time $2^{O(\ddim)}n\log n$,
the projection can be constructed even faster, in time $O(n\log n)$.
But even more importantly,
the projected data structure inherently provides hooks into the larger set $M$,
and these hooks are crucial for our goal of locating centers in $M$,
which is (generally) a much richer point set than $Q$.

\subsection{Related Work}
\label{sec:related}

Metrics with bounded doubling dimension are known to generalize
Euclidean metrics of fixed-dimension.
Below we briefly mention known algorithms achieving $(1+\veps)$--approximation
for the $p$-median and $p$-center problems in Euclidean spaces
of fixed dimension $D$.
These are only intended to be a crude analogy to our results,
possibly providing yet another perspective.
Often, different tradeoffs are possible
between the number of centers $p$ and the dimension $D$.
We do not discuss approximation algorithms for general metrics,
as these do not achieve $(1+\veps)$--approximation.

We start with the $p$-median problem.
Arora, Raghavan and Rao~\cite{ARR98} were the first to obtain
$(1+\veps)$--approximation,
via a divide-and-conquer approach based on quadtrees and dynamic programming.
This approach was later improved by Kolliopoulos and Rao~\cite{KR07}
and by Badoiu, Har-Peled, and Indyk~\cite{BHI02}.
Har-Peled and Mazumdar~\cite{HM04}
added another technique of finding coresets, and obtained running time
$O(n+p^{p+2}\veps^{-(2D+1)p}\log^{p+1}n\log^p\frac{1}{\veps})$.
Kumar, Sabharwal, and Sen~\cite{KSS10} showed a different approach,
based on finding centroid sets, that runs in time $2^{(p/\veps)^{O(1)}}n$.
These approaches were later combined by Chen \cite{Chen06},
who obtains improved runtime when the dimension $D$ is large.

For the $p$-center problem, Agarwal and Procopiuc~\cite{AP02}
obtain $(1+\veps)$--approximation in time
$O(n\log p) + (p/\veps)^{O(D p^{1-1/D})}$.
Badoiu, Har-Peled, and Indyk~\cite{BHI02}
show an algorithm that runs in time $p^{O(p/\veps^2)}Dn$.

\subsection{Preliminaries}
\label{sec:prelims}

Let $(M,d)$ be a finite metric space.
The \emph{doubling dimension} of $M$, denoted $\ddim=\ddim(M)$,
is the smallest $k>0$ such that every ball (in $M$) can be covered
by $2^k$ balls of half the radius.
We denote the \emph{diameter} of the metric by $\diam(M):=\max_{x,y\in M} d(x,y)$,
and its \emph{aspect ratio} (or \emph{spread}) by
$ \Delta(M):=\frac{\max\aset{d(x,y):\ x,y\in M} }{ \min\aset{d(x,y):\ x\neq y\in M} }$.

Let $r>0$. An \emph{$r$-net} of a point set $S\subset M$
is a subset $N\subseteq S$ satisfying:
(a) packing property: for all ${x,y\in N}$ we have $d(x,y)\ge r$; and
(b) covering property: for all ${x\in S}$ we have $d(x,N)< r$.
\footnote{Another common definition has a strict inequality in condition (a)
rather than in (b).
Our analysis can be adapted to this definition by changing constants.
}
Such a net always exists, and can be constructed greedily
by considering the points one by one in an arbitrary order.

\section{Our Data Structure}\label{sec:DS}

\paragraph{The Net Hierarchy.}

Our data structure is based on a lot of previous work on algorithms
and data structures for doubling metrics,
in particular for nearest neighbor search \cite{KL04,HM06,BKL06,CG06}.
But despite the overall similarity, some technical details
differ slightly from each of those papers.
Let $i_\topp:= \ceil{\log_2\diam(M)}$, and assume for simplicity
that the minimum interpoint distance in $M$ is $\min_{x\neq y\in M} d(x,y) = 1$
(otherwise we need to introduce $i_{\mathrm{bot}}$ as its logarithm).

Let $Y_0=M$, and for $i=1,\ldots,i_\topp$ let $Y_{i}$ be a $2^i$-net of $Y_{i-1}$.
Note that it is not necessarily a $2^i$-net of $M$,
but it does cover $M$ indirectly via the nets at lower levels.
We sometimes refer to $Y_i$ as the \emph{level} $i$ net.
By definition, $Y_i \subseteq Y_{i-1}$,
so when we refer to $y\in Y_i$ we mean the copy of $y$ which is in $Y_i$.
These nets form a natural hierarchy, with $Y_0$ being on the bottom,
and a singleton $Y_{i_\topp}=\aset{y_\topp}$ at the top of the hierarchy.
This hierarchy may be represented by a directed acyclic graph $G_M$,
whose vertex set is the union of all the nets $Y_i$
(so a point $y\in M$ may have multiple copies in this graph),
and with an arc from every $y_i\in Y_i$ to every $y_{i-1}\in Y_{i-1}$ for which $d(y_i,y_{i-1})\leq 2^{i}$.
We prefer not to maintain the graph $G_M$ explicitly;
instead, our data structure has two main components,
a tree $T$ and a collection of $c$-lists, which are defined below.

\paragraph{The Tree $T$.}

The hierarchy is represented by a tree that is defined as follows.
First construct $G_M$ as explained above. Next, every node in $G_M$ keeps only one of its incoming arcs
that is chosen arbitrarily except for giving higher priority
to the arc coming from another copy of the same point of $M$ (if it exists).
The surviving arcs define (when ignoring the edge orientations)
a tree, denoted $T=T_M$, which is rooted at $y_\topp$.
Because of the prioritization rule,
whenever a point $y\in Y_{i}$ has only one child in the tree $T$,
this child must correspond to the same point $y$ but in $Y_{i-1}$.
Thus, every non-branching path in $T$ consists of
copies of the same point in $M$ in consecutive nets.
By contracting each such path while recording
the range of nets in which it participates,
we can store the tree $T$ more compactly, using only $O(m)$ space
(recall the tree has $m$ leaves).
However, as explained a bit later,
we actually employ a more limited compaction,
that results with a weaker space bound $2^{O(\ddim)}m$.

We supplement $T$ with a data structure that supports constant-time
lowest common ancestor (LCA) queries using an additional $2^{O(\ddim)}m$ words
\cite{HT84} (see also~\cite{BF00} for a simplified version).
For the $1$-center and $p$-center algorithms, we supplement $T$ also
with a data structure for weighted level ancestor queries \cite{FM96,KL07},
which locate an ancestor of $q\in M$ at level $i$ (i.e., in $Y_i$) in $O(\log \log \log \Delta)$ time. The preprocessing for the weighted level ancestor queries requires $2^{O(\ddim)}m\log \log \log \Delta$ time.

\paragraph{The $c$-Lists.}
For some constant $c\ge1$ that will be determined later,
we maintain for every net point $y\in Y_i$ a so-called \emph{$c$-list}
$$L_{y,i,c}:=\{z\in Y_{i-1}: d(y,z)\leq c\cdot 2^i\}.$$
The $c$-lists allow us to traverse the ball of radius $c2^i$ in the next level of the hierarchy. If $c=1$, this list can be viewed as the set of arcs leaving $y\in Y_i$ in $G_M$.
When $c\ge 1$, these lists can be used (via straightforward filtering)
to recover the arcs of $G_M$.
Since $c$ is an absolute constant,
the size of each $c$-list is at most $c^{O(\ddim)}\le 2^{O(\ddim)}$
(see e.g.~\cite{GKL03,KL04}).
We do not store the $c$-list explicitly for every point in every net,
as this might require too much space.
We say that a $c$-list of a point $y\in Y_i$ is trivial if it has size $1$,
in which case the only point in this list must be the copy of $y$ in $Y_{i-1}$.
We store only nontrivial $c$-lists,
the number of which is at most $2^{O(\ddim)}m$~\cite[Theorem 2.1]{KL04}.
It follows that the total space usage for the $c$-lists is $2^{O(\ddim)}m$.

The nontrivial $c$-lists also limit the compaction of the tree $T$ as follows.
We compact $T$ only along paths whose nodes
are both non-branching and have trivial $c$-lists.
By the above bound on the number of nontrivial $c$-lists,
our limited compaction of $T$ uses at most $2^{O(\ddim)}m$ space.

\paragraph{Preprocessing time.}
The preprocessing stage first employs the data structure of~\cite{KL04}
to construct the $c$-lists in $2^{O(\ddim)}m\log \Delta \log\log \Delta$ time,
by simply inserting the data points one after the other.
We then scan this structure, from top to bottom,
to introduce direct pointers as dictated by the $c$-lists
(i.e., from $L_{y,i,c}$ to relevant $L_{z,i-1,c}$),
and also construct the tree $T$ in its compacted version.
This entire process takes $2^{O(\ddim)}m\log \Delta \log\log \Delta$ time.

\paragraph{Projected Tree.}

A key tool in getting faster runtime is a projection of the tree $T=T_M$
onto a subset of points $Q\subseteq M$.
The idea is to consider the subtree of $T$ induced by the leaves that are points in $Q$. We will denote this projected tree $T|_Q$.
Observe that this projected tree might be very different from the tree $T_Q$
that would be constructed for $Q$ independently of $M$;
in particular, the latter cannot contain points from $M\setminus Q$.
In the interest of runtime, we maintain the projected tree somewhat implicitly;
what the data structure stores explicitly is a compacted version,
in which all non-branching paths of $T|_Q$ are contracted,
and this clearly uses only $O(n)$ space.
Notice that such a contracted path of $T|_Q$
might contain nodes that \emph{are} branching in $T$,
which possibly correspond to distinct data points in $M$.
Although we have only the compacted version of $T|_Q$ at hand,
we can implement a traversal down the \emph{un-compacted} tree $T|_Q$,
as described in Lemma~\ref{lemma:traversal of uncompacted_projection}.

To aid in the construction of the projected tree, we number the leaves of $T$ in depth-first search (DFS) order.
In addition, for every node $u\in T|_Q$ we denote by $\wt(u)$ the number of leaves in its subtree. We can compute the weight of all the nodes in $T|_Q$ in time $O(n)$ by a simple scan.

\begin{lemma}\label{lem:TQ}
When a query $Q$ is given, the compacted version of $T|_Q$ can be computed in time $O(n\log n)$.
\end{lemma}
\begin{proof}[Proof (Sketch)]
To create $T|_Q$, first sort $Q$ according to the DFS numbering. Notice that the order in which points from $Q$ are encountered when performing a DFS on $T$ is exactly the order in which they would be encountered had we performed a DFS on $T|_Q$. Hence the sorted $Q$ gives us this order. We now use LCA queries to simulate the DFS on $T|_Q$, in order to construct $T|_Q$. This is done as follows. Denote by $Q_i\subseteq Q$ the first $i-1$ points in the ordered Q. We scan $Q$ by the DFS order, and when we reach the $i^{th}$ point, say $q_i\in Q$, we assume we have already constructed $T|_{Q_{i-1}}$ on the first $i-1$ points of $Q$. We now wish to insert $q_i$ to this tree to obtain $T|_{Q_i}$. To do this, compute $u=\text{LCA}(q_{i-1},q_i)$. This node $u$, which has to be part of $T|_{Q_i}$, is either on the path from the leaf corresponding to $q_i$ to the root of $T|_{Q_{i-1}}$, or is an ancestor of the root of $T|_{Q_{i-1}}$.
To locate its exact position, we traverse $T|_{Q_{i-1}}$
from the leaf corresponding to $q_{i-1}$ upwards towards the root,
testing at each node $v$ if this is the location into which $u$ should be inserted. The testing at $v$ is performed via an LCA query between $v$ and $u$. If the LCA query returns $u$, then the traversal needs to continue. If not, then $u$ is inserted as a child of $v$, either breaking an edge or inserting a new leaf. The entire process simulates the DFS search on $T|_Q$ and hence takes $O(n)$ time.
\end{proof}

The next lemma is used to traverse the un-compacted tree $T|_Q$
while using the data structure of its compacted version.

\begin{lemma}\label{lemma:traversal of uncompacted_projection}
Given the compacted $T|_Q$, a node $v$ in the un-compacted version of $T|_Q$ together with its weight $\wt(v)$, and node $w$ which is the closest descendant of $v$ in the compacted $T|_Q$ (and could possibly be $v$ itself),
it is possible to locate the children of $v$ in the un-compacted $T|_Q$, together with their weights,
in time $2^{O(\ddim)}$.
\end{lemma}
\begin{proof}
Suppose first that $v$ is a branching node in $T|_Q$.
For each child $u$ of $v$ in the compacted tree $T|_Q$,
we find the respective child of $v$ in the un-compacted $T|_Q$ as follows:
Run an LCA query between $u$ and every child of $v$ in $T$.
All of those queries will return $v$, except for one query that will return
the required child of $v$ in the un-compacted tree
(the one that is also an ancestor of $u$).
The total time for all such queries is $2^{O(\ddim)}$.

Suppose next that $v$ is a non-branching node,
and hence is not a part of the compacted version of $T|_Q$.
We find the child of $v$ that is an ancestor of $w$ in the un-compacted tree
as follows:
Perform an LCA query between $w$ and each of $v$'s children in $T$. All of those queries will return $v$, except for one query that will not return $v$,
but rather the child of $v$ that is also an ancestor of $w$, denoted by $w'$.
Notice that in this case, $w$ is also the closest descendant of $w'$ in the compacted $T|_Q$, which is needed to continue our traversal and proceed to $w'$.
\end{proof}

\paragraph{Standard Operations on the Net Hierarchy.}

A basic operation in a net hierarchy is a \emph{recursive scan},
where given a point $y_i\in Y_i$,
we scan its $c$-lists and apply the same procedure recursively on these points.
During this process, we discard any duplicates we find
(e.g., if we reach the same point in $Y_{i-2}$ via different points in $Y_{i-1}$).

\begin{definition}\label{def:c_descendant}
Let $y\in Y_i$.
A point $x\in M$ is called a \emph{$c$-list-descendant} of $y$
if it can be reached from $y$ using a recursive scan of the $c$-lists.
We then also say that $y$ is a \emph{$c$-list-ancestor} of $x$.
\end{definition}

\begin{lemma}\label{lemma:net_c_descendant_distance}
Let $y\in Y_i$, and let $x\in Y_j$ be a $c$-list-descendant of $y$.
Then $d(x,y)\leq c2^{i+1} - c2^{j+1}$.
\end{lemma}
\begin{proof}
The proof is by induction on $i-j$. The base case is trivial. For the inductive step, for every $y\in Y_i$, the distance between $y$ and any of the points in its $c$-list is at most $c2^i$. For every $x\in Y_j$ which is a $c$-list-descendant of $y$, there exists a $y_{i-1}\in L_{y,i,c}$ such that $x$ is a $c$-list-descendant of $y_{i-1}$. Therefore, $d(x,y)\leq d(x,y_{i-1}) + d(y_{i-1},y)\leq c2^i - c2^{j+1} + c2^i = c2^{i+1}  - c2^{j+1}$.
\end{proof}

Notice that a point $x$ can be a $c$-list-descendant of $y$
even if in the tree $T$ it is not a descendant of $y$.
However, ancestors and descendants in $T$ also have bounds on the distance between them.
\begin{lemma}\label{lemma:ancestor_descendant_distance_T}
Let $y$ be an ancestor of $x$ in $T$, such that $y\in Y_i$ and $x\in Y_j$, where $i>j$. Then
$$d(x,y)\leq 2^{i+1} - 2^{j+1} < 2^{i+1}.$$
\end{lemma}

\begin{proof}
The distance between a parent from $Y_i$ and its child in $T$ is at most $2^i$. Therefore, by summation on the path from $y$ to $x$ in $T$, and the triangle inequality $d(x,y)\leq \sum_{k=j+1}^{i} 2^k = 2^{i+1} - 2^{j+1}$.
\end{proof}

The following lemma is crucial to searching the vicinity of a given point
with some refinement factor $\veps>0$,
by executing a recursive scan with limited depth.
This process will be used several times in our various algorithms.
\begin{lemma}[Descendents Search with Refinement $\veps$]
\label{lemma:descendant_search}
Let $y\in Y_i$ and $x\in M$ be such that $d(x,y)\leq 2^{i}$,
and suppose $c\ge3$.
Then for every refinement constant $0<\veps\le1/2$,
a recursive scan of $c$-lists that stops at level ${i-\log({1}/{\veps})}$ will traverse
a point $x'\in Y_{i-\log({1}/{\veps})}$ for which $d(x,x')\leq \veps 2^i$.
In addition, the number of points traversed in such a scan is at most
$\veps^{-O(\ddim)}$.
\end{lemma}
\begin{proof}
Let $x_{\veps}$ be the ancestor of $x$ in $T$ who is in the $Y_{i-\log({1}/{\veps})}$ net, and so by Lemma~\ref{lemma:ancestor_descendant_distance_T} we have that $d(x,x_{\veps})\leq 2^{i-\log \frac{1}{\veps}+1}$. We prove by induction that for every $ i-\log \frac{1}{\veps}+1 \leq j \leq i$ and every $y_{j} \in Y_{j}$ such that $d(x,y_j)\leq 2^{j+1}$, the recursive scan of $c$-lists from $y_{j}$ will reach $x_\veps$. This will suffice as for every $x$ such that $d(x,y)\leq 2^i$, we also have $d(x,y)\leq 2^{i+1}$. For the base case, $j = i-\log \frac{1}{\veps} +1$, and so $d(y_j,x_{\veps})\leq d(y_j,x) + d(x,x_{\veps}) \leq 2^{j+1} +2^j \leq c2^j$, and so $x_{\veps}$ is in the $c$-list for $y_j$.

For the induction step, assume that the claim is correct for $j-1$. Consider $x_{j-1}\in Y_{j-1}$ which is the ancestor of $x_{\veps}$ in $T$, and therefore is also an ancestor of $x$. Then by Lemma~\ref{lemma:ancestor_descendant_distance_T}, $d(x,x_{j-1})\leq 2^{j}$ and by the induction hypothesis, a recursive scan on the $c$-lists starting from $x_{j-1}$ will reach $x_{\veps}$. Then $d(x_{j-1},y_{j})\leq d(x_{j-1},x)+d(x,y_{j}) \leq 2^{j} + 2^{j+1}\leq c2^{j}$, and so a recursive scan on the $c$-lists starting from $y_{j}$ must go through $x_{j-1}$ and eventually reach $x_{\veps}$.

The number of points traversed can be bounded as follows. Each point $x$ not in $Y_{i-\log({1}/{\veps})}$ that is encountered needs to scan its $c$-list which is of size $2^{O(\ddim)}$. So at $k$ levels beneath $y$ we scan at most $2^{O(\ddim\cdot k)}$ points. The last level scanned is when $k=\log \frac{1}{\veps}$, so using a geometric series we obtain that the total number of points scanned is $2^{O(\ddim \cdot \log({1}/{\veps}))} = \veps^{-O(\ddim)}$.
\end{proof}

\section{A Simple Algorithm for 1-median}
\label{sec:simple1med}

In this section we provide a simple iterative algorithm for 1-median,
for the purpose of explaining the basic approach used by
our main result for 1-median in Section \ref{sec:eff1med}.
This basic approach is similar, but not identical, to the known
algorithms for NNS \cite{KL04,HM06,BKL06,CG06}.
We remark that there is a well-known randomized algorithm that achieves
an (expected) $2$--approximation for 1-median
by picking a random point from $Q$ to be the center.
Below, we present a \emph{deterministic} 6--approximation algorithm,
which has the advantage that it is then easily refined
to achieve $(1+\veps)$--approximation.
Unlike that randomized algorithm, ours can probably be adapted
to the case where $Q$ need not be a subset of $M$
(or alternatively, when the center must come from $M\setminus Q$).

\begin{theorem}
There is an algorithm that preprocesses a finite metric $M$
in $2^{O(\ddim)}m\log \Delta(M)\log\log \Delta(M)$ time using $2^{O(\ddim)}m$ space,
so that subsequent $1$-median queries on a set $Q\subseteq M$,
can be answered within $(1+\veps)$--approximation (for any desired $0<\veps\le \tfrac12$) in time $n(2^{O(\ddim)}\log \Delta + \veps^{-O(\ddim)})$.
\end{theorem}

The preprocessing algorithm simply builds the net hierarchy for the metric $M$
(see Section \ref{sec:DS}).
The query algorithm is described in Figure \ref{fig:1medSimple}.
For convenience, we use the shorthand $\median(y,Q)$ for $\median(\aset{y},Q)$.
By convention, for all $i<0$ we define $Y_i:=M$ (similarly to $Y_0$),
and note that the corresponding $c$-lists can be computed on the fly by
a direct filtering of the respective $c$-list at level $0$.

\begin{figure*}[h]
\hrule\medskip
\begin{tabbing}
\hspace*{.02in} \= \hspace*{.25in} \= \hspace*{.10in} \= \hspace*{.25in} \= \hspace*{.25in} \= \hspace*{.25in} \=\kill
\> 1. {\bf let} $y\leftarrow y_\topp$\\
\> 2. {\bf for each} $i$ from $i_\topp$ down to $-3$\\
\> 3. \> {\bf let}  $\hat{y}\leftarrow \argmin_{z\in L_{y,i,7}} \median(Q,z)$\\
\> 4. \> {\bf if}   $\median(Q,\hat{y}) > 3n \cdot 2^{i-1}$ {\bf then return} $y$.\\
\> 5. \> {\bf else} $y\leftarrow \hat{y}.$\\
\> 6. {\bf return} $y$.
\end{tabbing}
\vspace{-0.3in}
\caption{Simple algorithm for $1$-median query on a set $Q\subset M$}
\label{fig:1medSimple}
\medskip\hrule
\end{figure*}

\paragraph{Correctness Analysis.}
Assume for now that the algorithm returns from line 4.
(We discuss later the more special case where the algorithm reaches line 6.)
For the following, let $\ifin$ be the final value of $i$ (i.e., at line 4),
and let $y$ and $\hat{y}$ refer to their values in the algorithm
at the end of the execution.
It can be verified that the condition in line 4 must fail at least once,
by considering $y_\topp$ as a potential center $\hat y$,
and bounding the distance between every point in $Q\subset M$ to $y_\topp$
using Lemma \ref{lemma:ancestor_descendant_distance_T}.
Therefore,
\begin{equation}
\textstyle
  \label{eq:fin}  
  \sum_{q\in Q} d(q,\hat{y})>3 n\cdot 2^{\ifin-1},
  \quad\text{ and }\quad
  \sum_{q\in Q} d(q,y)\leq 3n\cdot 2^{\ifin}.
\end{equation}
Let $a\in M$ be an optimal solution to the $1$-median problem on $Q$.
Let $a_{\ifin-1}\in Y_{\ifin-1}$ be an ancestor of $a$ in $T$.
Then $d(a_{\ifin-1},a)\leq 2^{\ifin}$
by Lemma \ref{lemma:ancestor_descendant_distance_T}.

\begin{lemma}\label{lemma:a_is_near}
$d(a_{\ifin-1},y)\leq 7\cdot 2^{\ifin}$, and thus $a_{\ifin-1} \in L_{y,\ifin,7}$.
\end{lemma}
\begin{proof}
Using the triangle inequality, the optimality of $a\in M$, and then \eqref{eq:fin},
\begin{align*}
n\cdot d(a_{\ifin-1},y)
\leq \sum_{q\in Q} \left[d(a_{\ifin-1},a) + d(a,q) + d(q,y)\right]
\leq \sum_{q\in Q} d(a_{\ifin-1},a) + 2\sum_{q\in Q} d(q,y)
\leq 7\cdot n2^{\ifin}.
\qquad\mbox{}
\qedhere
\end{align*}
\end{proof}

\begin{lemma}\label{lem:simpleLB}
$\sum_{q\in Q} d(q,a) > n\cdot 2^{\ifin-1}$.
\end{lemma}
\begin{proof}
By Lemma \ref{lemma:a_is_near},
when the algorithm computes $\hat{y}$ in the final iteration,
one of the options it considers is $a_{\ifin-1}$,
and so 
\begin{align*}
  \sum_{q\in Q} d(q,\hat{y})
 \leq \sum_{q\in Q} d(q,a_{\ifin-1})
 \leq \sum_{q\in Q} \left[d(q,a) + d(a,a_{\ifin-1})\right]
 \leq \sum_{q\in Q} d(q,a) + n\cdot 2^{\ifin}.
\end{align*}
Combining this with \eqref{eq:fin} and rearranging, the lemma follows.
\end{proof}
Thus, if we returned from line 4, then using \eqref{eq:fin},
the approximation factor achieved is
$
  \frac{\sum_{q\in Q} d(q,y)}{\sum_{q\in Q} d(q,a)}
  < \frac{3\cdot n\cdot 2^{\ifin}}{n\cdot 2^{\ifin-1}}
  = 6
$.
If we returned from line 6, Lemma~\ref{lemma:a_is_near} holds also for $i=-3$,
and thus at the last execution of line 3, we have $d(a,y) \le \frac{7}{2^3}$.
But since there cannot be two points with distance less than $\frac{7}{2^3}<1$,
we see that $y=a$, and the returned point is an optimal solution $a$.
We remark that a similar effect can be achieved by stopping at $i=0$,
possibly increasing the value of $c$.

\subsection{Refinement to $(1+\veps)$--approximation}

We now improve the approximation factor to $1+\veps$
for an arbitrary $\veps\in(0,\frac12]$.
We can utilize the fact that $a$ is a descendant of $a_{\ifin-1}$ in $T$, so $d(a,a_{\ifin-1}) \leq 2^\ifin$, and that $a_{\ifin-1}\in L_{y, \ifin, c}$. As such, we perform a descendant search with refinement constant ${\veps}/{2}$,
starting from each member of $L_{y, \ifin, c}$.
By Lemma~\ref{lemma:descendant_search}, we are guaranteed to traverse
a point $a_{\frac{\veps}{2}}$ such that $d(a,a_{\frac{\veps}{2}})\leq \frac{\veps}{2} 2^{\ifin}$. For each point $x$ traversed in this process, we compute $\median(Q,x)$,
and eventually report a center candidate $x$ with minimal objective value
$\median(Q,x)$.
Using \eqref{eq:fin} again, this objective value is
\begin{align*}
\median(Q,x)
  \leq \median(Q,a_{\frac{\veps}{2}})
  \leq  \sum_{q\in Q} \left[d(q,a) + d(q,a_{\frac{\veps}{2}}) \right]
  \leq \median(Q,a) + \tfrac{\veps}{2} n2^\ifin
  \leq (1+\veps) \median(Q,a).
\end{align*}

\subsection{Runtime Analysis}
The running time of the first part of the algorithm is $2^{O(\ddim)}n\log \Delta$,
as there are at most $O(\log \Delta)$ levels,
and at each level we compute the distance from every point in $Q$
to every point $z\in L_{y,i,7}$.
In the second part of the algorithm (the descendants search) we compute the cost of each of the $\veps^{-O(\ddim)}$ center candidates in $O(n)$ time.
The total runtime is $n(2^{O(\ddim)}\log \Delta + \veps^{-O(\ddim)})$,
and the space usage is just that of the hierarchy, which is $2^{O(\ddim)}m$.

\section{An Efficient Algorithm for $1$-median}
\label{sec:eff1med}

\begin{theorem}\label{thm:eff1med}
There is an algorithm that preprocesses a finite metric $M$ of size $m$
in time $2^{O(\ddim)}m\log \Delta(M)\log\log \Delta(M)$ using $2^{O(\ddim)}m$ memory words,
so that subsequent $1$-median queries on a set $Q\subseteq M$ of size $n$
can be answered within approximation factor $1+\veps$
(for any desired $0<\veps\le1/2$)
in time $O(n\log n) + 2^{O(\ddim)}\log \Delta(M) + \veps^{-O(\ddim)}$.
\end{theorem}

This theorem builds on the simple algorithm from Section \ref{sec:simple1med},
refining the approach therein using two main ideas.
First, as we iterate down the levels $i$,
some query points $q\in Q$ might get further away from the current center $y_i\in Y_i$.
But then, picking any $c$-list-descendant of $y_i$ as the final center
will give approximately the same contribution from those far query points.
This speeds up the traversal down the hierarchy as query points
need not be considered once they get far enough from $y_i$.
The second idea is to cluster query points that are close to each other,
relative to the current level $i$, into one (weighted) representative point.
This (crude) clustering must be computed quickly,
and indeed it is achieved using the projection tree $T|_Q$.
Once we bound the number of weighted representatives under consideration
in each iteration, we obtain a significant speedup.

\paragraph{Algorithm Description.}

We first describe a constant factor approximation algorithm,
which is detailed in Figure~\ref{fig:1medEffecient}
using $\alpha,c'>0$ to denote sufficiently large constants.
Similarly to the simple algorithm in Section \ref{sec:simple1med},
the algorithm iterates (in lines 3--12) down the levels $i$,
while maintaining a candidate center $y_i\in Y_i$.
However, the iterations here start at the root of $T|_Q$ (instead of at $y_\topp$).
Observe that the next candidate $y_{i-1}$ is always chosen
from the $c$-list of $y_i$ (lines 9,12).

\begin{figure*}[h]
\hrule\medskip
\begin{tabbing}
\hspace*{.02in} \= \hspace*{.25in} \= \hspace*{.10in} \= \hspace*{.25in} \= \hspace*{.25in} \= \hspace*{.25in} \=\kill
\> 1 . {\bf compute} $T|_Q$\\
\> 2 . {\bf let}
  $\iroot \leftarrow \levell(\roott(T|_Q))$;
  $y_{\iroot-1}\leftarrow \roott(T|_Q)$;
  $R_{\iroot}\leftarrow \aset{\roott(T|_Q)}$;
  $\summ\leftarrow 0$ \\
\> 3 . {\bf foreach} $i$ from $\iroot-1$ down to $0$\\
\> 4 . \> \> {\bf let} $R_i\leftarrow \emptyset$\\
\> 5 . \> \> {\bf foreach} $r\in R_{i+1}$\\
\> 6 . \> \> \> \> {\bf if} $d(r,y_i)> c'\cdot 2^i$ \\
\> 7 . \> \> \> \> {\bf then let} $\summ \leftarrow \summ + \wt(r)\cdot d(r,y_i)$\\
\> 8. \> \> \> \> {\bf else  let} $R_i \leftarrow R_i \cup \{\mbox{children of $r$ in non-compacted $T|_Q$} \}$.\\
\> 9. \> \> {\bf let} $\hat{y}\leftarrow \argmin_{x\in L_{y_i,i,c}} \{\sum_{r\in R_i} d(r,x)\cdot \wt(r)\}$\\
\> 10. \> \> {\bf if} $\summ+\sum_{r\in R_i} d(r,\hat{y})\cdot \wt(r)>\alpha \cdot n \cdot 2^{i-1}$\\
\> 11. \> \> {\bf then return} $y_i$\\
\> 12. \> \> {\bf else} $y_{i-1}\leftarrow \hat{y}$\\
\> 13. {\bf return} $y_{-1}$
\end{tabbing}
\vspace{-0.25in}
\caption{Efficient Algorithm for $1$-median query on a set $Q\subset M$}
\label{fig:1medEffecient}
\medskip\hrule
\end{figure*}

During the iterations, the algorithm maintains also a set $R_i$ of representatives
to some points of $Q$, those points that are not too far, as explained next.
The level $i$ representative of a point $q\in Q$, denoted $r_i(q)$,
is the (unique) ancestor $r\in Y_i$ of $q$ in $T|_Q$.
Notice that this is the same ancestor as in the tree $T$.
The algorithm also uses, for each representative $r\in R_i$,
a weight denoted $\wt(r)$, which is the number of points in $Q$
that have $r$ as an ancestor in $T|_Q$.
This weight is calculated for each node in $T|_Q$ during
the tree's construction in line 1.
The set of representatives $R_i$ is constructed (in lines 4,8)
from children of $R_{i+1}$ in $T|_Q$,
which clearly maintains the invariant $R_i\subset Y_i$.
In this process, we skip (via the condition in line 6)
representatives $r\in R_{i+1}$ that are far enough from $y_i$,
in which case we add their weighted distance $\wt(r)\cdot d(r,y_i)$
to a variable called $\summ$.
The purpose of this variable is to accumulate all those weighted distances,
but note that each weighted distance is taken relative to $y_i$
at the iteration in which the representative $r$ fails the condition in line 6.
Denote by $\summ_i$ the value of variable $\summ$ at the end of iteration $i$.
For representatives $r\in R_{i+1}$ that are close enough to $y_i$,
we need to compute their children in the un-compacted $T|_Q$ (in line 8).
For simplicity sake, the algorithm's description assumes that the tree $T|_Q$
is available in its un-compacted version.
The necessary operations can be implemented using the data structure
for the compacted version
by Lemma~\ref{lemma:traversal of uncompacted_projection}.

\subsection{Correctness Analysis}\label{app:eff1med_correctness}
We say a point $q\in Q$ is \emph{far} at level $i$ if it has no representative
in $R_i$,
which means that during some iteration $i'>i$ its representative was skipped.
A point $q\in Q$ is \emph{near} if it is not far.
Let $F_i$ denote the points of $Q$ that are far at level $i$,
and similarly $N_i=Q\setminus F_i$ for the points that are near.
Notice that $F_{i} \supseteq F_{i+1}$ and $N_{i} \subseteq N_{i+1}$.

Let $\ifin$ be the value of $i$ at the end of the execution.
This is the ``last'' level (time-wise) considered by the algorithm,
and the analysis shall rely on the corresponding partition $Q=N_\ifin\cup F_\ifin$.
For $q\in Q$, we denote its representative in $R_i$ by $r_i(q)$.
We let $r_q$ be the ``last'' representative of $q$, formally defined as follows.
If $q\in F_\ifin$, define $i_q$ as the smallest $i$ such that $q\in N_i$.
Intuitively, this is the ``last level'' in which $q$ has a representative,
and also the (unique) value of $i$ such that
$q\in N_i\setminus N_{i-1}=F_{i-1}\setminus F_i$
(assuming by convention $N_{\ifin-1}=\emptyset$ and $F_{\ifin-1}=Q$).
Otherwise (i.e., $q\in N_\ifin$), define $i_q:=\ifin$.
In both cases, let $r_q:=r_{i_q}(q)$. Notice that $r_q\in Y_{i_q}$.

At iteration $i$, the variable called $\summ$ receives (in line 7)
a contribution for every point $q\in F_i\setminus F_{i+1}$.
Observe that this contribution is proportional to $d(r_{i+1}(q),y_i)$,
and the last representative of $q$ is at level $i_q=i+1$.
Hence, $r_q = r_{i_q}(q) = r_{i+1}(q)$, and by the condition in line 6,
\begin{equation}
  \label{eq:drq}
  d(r_q,y_{i_q-1})
  = d(r_{i+1}(q),y_{i})
  > c'2^{i}
  = c'2^{i_q-1}.
\end{equation}
Summing the aforementioned contributions over all iterations up to $i$,
we see that
\begin{equation}
  \label{eq:sumi}
\summ_i
  = \sum_{i'\ge i}\ \sum_{q\in F_{i'}\setminus F_{i'+1}} d(r_{i'+1}(q),y_{i'})
  = \sum_{q\in F_i} d(r_q,y_{i_q-1}).
\end{equation}

In addition, $r_q\in Y_{i_q}$ and is an ancestor of $q$ in $T$.Thus, by Lemma~\ref{lemma:ancestor_descendant_distance_T},
\begin{align}
  \label{eq:qinQ}
  &\forall q\in Q,\ \quad d(q,r_q)\leq 2^{i_q +1} \\
  \label{eq:qinN}
  &\forall q\in N_\ifin,\quad d(q,r_q)\leq 2^{\ifin+1}.
\end{align}
Below, $\hat y$ refers to its value at the end of the execution.

We assume from now on that the algorithm halts during some iteration
and returns the value from line 11.
Similarly to Section~\ref{sec:simple1med},
the special case where the algorithm returns from line 13 is proved
by replacing Eqn.~\eqref{eq:lastiter} and its consequences
with the fact that we reached iteration $i=0$.
Thus, at the last iteration, $i=\ifin$, the algorithm halts, and
\begin{align}
  \label{eq:lastiter}
  \summ_\ifin+\sum_{r\in R_\ifin} d(r,\hat{y}) \cdot \wt(r)
  &= \summ_\ifin+\sum_{q\in N_\ifin} d(r_q,\hat{y})
  > \alpha \cdot n \cdot 2^{\ifin-1}.
\intertext{Similarly, at the previous to last iteration
$i=\ifin+1$ and $\hat{y}$ is assigned $y_\ifin$, hence}
  \label{eq:prevlastiter}
  \summ_{\ifin+1}+\sum_{r\in R_{\ifin+1}} d(r,y_\ifin) \cdot \wt(t)
  &= \summ_{\ifin+1}+\sum_{q\in N_{\ifin+1}} d(r_{\ifin+1}(q),y_\ifin)
  \leq \alpha \cdot n \cdot 2^{\ifin}.
\end{align}
This inequality holds even in the special case where $\ifin=\iroot-1$
and there was no previous to last iteration.
Indeed, we have that $F_{\ifin +1} = \emptyset$, $\summ_{\ifin+1} =0$, $R_{i*+1} = \aset{\roott(T|_Q)}$ and $y_{\ifin} = \roott(T|_Q)$,
and therefore, $\summ_{\ifin+1} +\sum_{r\in R_{\ifin+1}} d(r,y_\ifin) = 0 \leq \alpha \cdot
n \cdot 2^{\ifin}$.

\begin{lemma}\label{lem:sumfinal}
$\summ_{\ifin}+\sum_{q\in N_{\ifin}} d(r_q,y_\ifin) \leq (\alpha+2) \cdot n \cdot 2^\ifin$.
\end{lemma}

\begin{proof}
We write the lefthand-side as
\begin{align*}
\summ_{\ifin} &+\sum_{q\in N_{\ifin}} d(r_q,y_\ifin) = \\
  &= \summ_{\ifin+1}+\sum_{q\in N_{\ifin+1}-N_\ifin}d(r_{\ifin+1}(q),y_\ifin)+\sum_{q\in N _{\ifin}}
d(r_\ifin(q),y_\ifin) \\
  &\leq \summ_{\ifin+1}+\sum_{q\in N_{\ifin+1}-N_ \ifin}d(r_{\ifin+1}(q),y_\ifin)+\sum_{q\in N_{\ifin}} \left[d(r_{\ifin}(q),r_{\ifin+1}(q))+d(r_{\ifin+1}(q),y_\ifin)\right] \\
  &\leq \summ_{\ifin+1}+\sum_{q\in N_{\ifin+1}}d(r_{\ifin+1}(q),y_\ifin) + n2^{\ifin+1},
\end{align*}
where the last inequality follows from $r_\ifin{(q)}$ being a child of $r_{\ifin+1}(q)$ in $T|_Q$.
The lemma then follows by plugging in Eqn.~\eqref{eq:prevlastiter}.
\end{proof}

For every $q\in F_\ifin$, we have by Eqn.~\eqref{eq:drq} that
$d(q,r_q)\leq 2^{i_q+1} \leq \frac{4}{c'}d(r_q,y_{i_q-1})$.
In addition, $d(y_{i_q-1},y_\ifin)\leq c2^{i_q} \leq \frac{2c}{c'}d(r_q,y_{i_q-1})$, because $y_\ifin$ is a $c$-list-descendant of $y_{i_q-1}$
and thus Lemma~\ref{lemma:net_c_descendant_distance} applies.
To simplify notation, define $\beta:= \frac{4}{c'} + \frac{2c}{c'}$
and notice it can be made an arbitrarily small positive constant by controlling $c'$.
For example, it is always possible to make $\beta = \frac 12$.
We can now show that with respect to the query points $F_\ifin$,
our estimate $\summ_\ifin$ is a good approximation
for the cost of picking $y_\ifin$ as the center.
\begin{align}
  \nonumber
  \sum_{q\in F_\ifin} d(q,y_\ifin)
  &\leq \sum_{q\in F_\ifin} \left[d(q,r_q)  + d(r_q,y_{i_q-1})  + d(y_{i_q-1},y_\ifin)\right] \\
  \label{eq:sumFifinal}
  &\leq (1+\tfrac{4}{c'} + \tfrac{2c}{c'}) \sum_{q\in F_\ifin} d(r_q,y_{i_q-1})
  = (1+\beta) \summ_\ifin.
\end{align}
In addition, we show that with respect to the query points $N_\ifin$,
the representatives give a good approximation as well.
\begin{align}
  \label{eq:sumNifinal}
  \sum_{q\in N_\ifin} d(q,y_\ifin)
  \leq \sum_{q\in N_\ifin} \left[d(q,r_q)+ d(r_q,y_\ifin)\right]
  \leq n2^{\ifin+1} + \sum_{q\in N_\ifin} d(r_q,y_\ifin).
\end{align}

Let $a\in M$ be an optimal solution to the $1$-median problem $Q$,
and let $a_{\ifin-1}\in Y_{\ifin-1}$ be an ancestor of $a$ in $T$.
Thus, $d(a_{\ifin-1},a)\leq 2^{\ifin}$.
We next prove that $y_\ifin$ is near $a_{\ifin-1}$,
and thus also near $a$ itself.

\begin{lemma}\label{lem:aclosey}
$d(a_{\ifin-1},y_\ifin)\leq c2^\ifin$ and therefore $a_{\ifin-1}\in L_{y_\ifin,\ifin,c}$.
\end{lemma}
\begin{proof}
We start with the lefthand-side multiplied by $n$
\begin{align*}
n\cdot d(a_{\ifin-1},y_\ifin)
&= \sum_{q\in Q} d(a_{\ifin-1},y_\ifin) \\
&\leq \sum_{q\in Q} [d(a_{\ifin-1},a) + d(a,q) +d(q,y_\ifin)] & \mbox{by triangle inequality}\\
&\leq n2^{\ifin} + 2\sum_{q\in Q} d(q,y_\ifin) & \mbox{by optimality of $a$}\\
& = n2^{\ifin} +2\Big[\sum_{q\in F_\ifin} d(q,y_\ifin)+\sum_{q\in N_\ifin} d(q,y_\ifin)\Big]\\
&\leq n2^{\ifin} +2\big[ (1+\beta) \summ_\ifin+   n2^{\ifin+1} + \sum_{q\in N_\ifin} d(r_q,y_\ifin) \big] & \mbox{by Eqns.~\eqref{eq:sumFifinal},\eqref{eq:sumNifinal}}\\
&\leq  5n2^{\ifin} +     2(1+\beta)(\alpha+2) \cdot n \cdot 2^\ifin & \mbox{by Lemma \ref{lem:sumfinal}}\\
&\leq cn2^\ifin.
\end{align*}
For the last inequality we need to pick a large enough constant $c>0$.
(Recall that $\beta$ can be made to be $\frac 12$ increasing $c'$ as needed.)
Dividing all by $n$ completes the proof.
\end{proof}

We now want to prove the guarantee of our approximation.
To this end we need an upper bound on the cost of the algorithm's solution,
which we establish by analyzing the stopping condition iteration.

\begin{lemma}\label{lem:UByhat}
$\sum_{q\in Q} d(q,y_\ifin) \leq [2+(1+\beta)(\alpha +2)] n2^\ifin.$
\end{lemma}
\begin{proof}
First, each $q$ is close to $r_q$ and thus
\begin{align*}\sum_{q\in N_\ifin} d(q,y_\ifin)
\leq \sum_{q\in N_\ifin} \left[d(r_q,q) + d(r_q,y_\ifin)\right]
\leq n2^{\ifin+1} + \sum_{q\in N_\ifin} d(r_q,y_\ifin).
\end{align*}
Thus,
\begin{align*}
\sum_{q\in Q} d(q,y_\ifin)
&= \sum_{q\in N_\ifin} d(q,y_\ifin) + \sum_{q\in F_\ifin} d(q,y_\ifin)\\
&\leq  n2^{\ifin+1} + \sum_{q\in N_\ifin} d(r_q,y_\ifin) + (1+\beta)\summ_\ifin \\
&\leq n2^{\ifin+1} + (1+\beta)(\summ_\ifin + \sum_{q\in N_\ifin} d(r_q,y_\ifin)) \\
&\leq   [2+ (1+\beta)(\alpha +2)]\ n2^\ifin.
\end{align*}
\end{proof}

\begin{lemma}\label{lem:sumiLB}
$\summ_\ifin \leq \frac{1}{1-\beta} \sum_{q\in F_\ifin} d(q,a)$.
\end{lemma}

\begin{proof}
First notice that

\begin{align*}
d(a,y_{i_q-1})
&\leq d(a,a_{\ifin-1}) + d(a_{\ifin-1},y_{i_q-1}) \\
&\leq 2^{\ifin} + c2^{i_q} - c2^{i_*}\\
& \leq c2^{i_q}\\
& \leq \frac{2c}{c'}d(r_q,y_{i_q-1}),
\end{align*}
where the bound on $d(a_{\ifin-1},y_{i_q-1})$ follows from Lemma~\ref{lemma:net_c_descendant_distance}. Therefore,
\begin{align*}
\sum_{q\in F_\ifin} d(q,a)
\geq \sum_{q\in F_\ifin} [d(r_q,y_{i_q-1}) - d(q,r_q) - d(a,y_{i_q-1}) ]
\geq  (1-\beta)\summ_\ifin.
\end{align*}
\end{proof}

We are now ready to provide a lower bound on the optimal solution.
Recall that $\hat y$ refers to its value at the end of the algorithm.
\begin{lemma}\label{lemma:bound_opt}
$\sum_{q\in Q} d(q,a) > (\alpha/{2} - 3)(1-\beta) n2^\ifin$
(assuming the algorithm returns from line 11).
\end{lemma}

\begin{proof}
\begin{align*}
\alpha \cdot n\cdot 2^{\ifin-1}
&< \summ_\ifin + \sum_{q\in N_\ifin}d(r_q,\hat{y}) & \mbox{\eqref{eq:lastiter}}\\
& \leq \summ_\ifin + \sum_{q\in N_\ifin}d(r_q,a_{\ifin-1}) & \mbox{by Lemma~\ref{lem:aclosey} and choice of $\hat y$}\\
&\leq \summ_\ifin + \sum_{q\in N_\ifin}[d(r_q,q) + d(q,a) + d(a,a_{\ifin-1})] \\
&\leq \summ_\ifin + n2^{\ifin+1} +\sum_{q\in N_\ifin}d(q,a) + n2^{\ifin}\\
&\leq   3 n2^\ifin +   \tfrac{1}{1-\beta}  \sum_{q\in F_\ifin} d(q,a) + \sum_{q\in N_\ifin}d(q,a) & \mbox{by Lemma~\ref{lem:sumiLB}}\\
& \leq 3 n2^\ifin + \tfrac{1}{1-\beta} \sum_{q\in Q} d(q,a).
\end{align*}
\end{proof}

We conclude that the algorithm achieves approximation factor
\begin{align*}
\frac{\sum_{q\in Q} d(q,y_\ifin)}{\sum_{q\in Q} d(q,a) }
\leq \frac{[2+(1+\beta)(\alpha +2)]n2^\ifin}{(1-\beta)({\alpha}/{2}-
3)n2^\ifin}
= \frac{[2+(1+\beta)(\alpha +2)]}{(1-\beta)({\alpha}/{2}- 3)}.
\end{align*}

\subsection{Refinement to $(1+\veps)$--approximation}

Our goal now is to improve the approximation factor to $1+\veps$ for arbitrary $\veps>0$. We can utilize the fact that $a$ is a descendant of $a_{\ifin-1}$ in $T$, so $d(a,a_{\ifin-1}) \leq 2^\ifin$, and that $a_{\ifin-1}\in L_{y_\ifin, \ifin, c}$. As such, we can perform a descendant search, as in Lemma~\ref{lemma:descendant_search}, starting from each member of $L_{y_\ifin, \ifin, c}$, with refinement constant $\veps'= \Theta(\veps)$.
By Lemma~\ref{lemma:descendant_search} we are guaranteed to traverse a point $a_{\veps'}$ such that $d(a,a_{\veps'})\leq \veps' 2^{\ifin}$. However, we wish to avoid the high runtime of computing the cost of each center candidate by summing the distances from all of $Q$ to that point. Instead, we once again speed up the process by removing far points, and using weighted representatives for the rest.

\paragraph{Speeding up the descendants search.}
Define the set of far points $F = \{q\in Q: d(q,y_\ifin) > \frac{3c}{2\veps'}2^\ifin\}$ for some $\veps'>0$ to be determined later.
The set of near points is $N:=Q\setminus F$. The points in $F$ are ignored in this phase of the algorithm.
For the points in $N$ we wish to find good representatives so that the number of representatives is few, and the additive
distortion caused by replacing the query points in $N$ with their representative is very small. To this end, consider the set of
representatives obtained as follows. Each point $q\in N$ is mapped to its ancestor in the compacted $T|_Q$ which is in $Y_k$ for the largest $k\leq {\ifin-\log ({1}/{\veps''})}$ for $\veps''>0$ to be determined later. Call this set of representatives $R_{\veps''}$, and give each $r\in R_{\veps''}$ a weight
$\wt(r)$ which is the number of points in $N$ that were mapped to $r$. Notice that the process of this mapping and weighting can be done
efficiently by scanning the compacted $T|_Q$ in $O(n)$ time.
Now, for each center candidate $x$ obtained by a descendants search from each of the points in  $L_{y_\ifin,\ifin,c}$ by using Lemma~\ref{lemma:descendant_search} with refinement constant $\veps''$, we compute $\sum_{r\in R_{\veps''}}d(r,x)\cdot \wt(r)$, and take the candidate which minimizes this cost.

We want to argue that the candidate returned is a $1+\veps$ approximation from the optimum. Denote this candidate by $x$. Notice that one of the candidates must be a point $a_{\veps''}$ which is an ancestor of $a$ in $T$ and is in $Y_k$ for some $k\leq {\ifin-\log (1/\veps'')}$. Therefore, $\sum_{r\in R_{\veps''}}d(r,x)\cdot \wt(r) \leq \sum_{r\in R_{\veps''}}d(r,a_{\veps''})\cdot \wt(r)$.

\begin{lemma}\label{lem:1med_eps}
$\sum_{q\in Q} d(q,x) \leq (1+\veps)\sum_{q\in Q} d(q,a)$.
\end{lemma}
\begin{proof}
Denote by $x_{\ifin-1}$ the $c$-list-ancestor of $x$ in $L_{y_\ifin, \ifin, c} \subseteq Y_{\ifin-1}$. First, for every $q\in F$,
\begin{align*}
d(q,x)
&\leq d(q,y_\ifin) +d(y_\ifin,x_{\ifin-1})+ d(x_{\ifin-1},x)  \\
&\leq d(q,y_\ifin) + c2^{\ifin-1}  + c2^\ifin & \mbox{by $x_{\ifin-1}\in L_{y_\ifin, \ifin, c}$ and Lemma \ref{lemma:net_c_descendant_distance}}\\
&\leq d(q,y_\ifin) + \veps' d(q,y_\ifin) & \mbox{since $q\in F$}\\
&= (1+\veps')d(q,y_\ifin),
\end{align*}
and similarly,
\begin{align*}
d(q,y_\ifin)
&\leq d(q,a) + d(a,a_{\ifin-1}) +d(a_{\ifin-1},y_\ifin) \\
&\leq d(q,a) + 2^\ifin +c2^\ifin & \mbox{by Lemma \ref{lem:aclosey}}\\
&\leq d(q,a) + \veps' d(q,y_\ifin).  & \mbox{since $q\in F$}
\end{align*}
Therefore, $d(q,y_\ifin) \leq \frac{d(q,a)}{1-\veps'}$.
Combining this with our earlier inequality, we get
\begin{align}
  \label{eq:farpts}
  \sum_{q\in F} d(q,x) \leq \frac{1+\veps'}{1-\veps'} \sum_{q\in F} d(q,a).
\end{align}
For the near points, we have
\begin{align}
\label{eq:nearpts}
\sum_{q\in N} d(q,x)
\leq \sum_{q\in N} d(q,a_{\veps''})
\leq \sum_{q\in N} \big[d(q,a) + d(a,a_{\veps''})\big]
\leq \sum_{q\in N} \big[d(q,a) + \veps'' 2^{\ifin+1}\big].
\end{align}
Altogether, the cost of the reported center candidate $x$ is
\begin{align*}
\sum_{q\in Q} d(q,x)
&\leq \frac{1+\veps'}{1-\veps'}\sum_{q\in F}d(q,a) + \sum_{q\in N} [d(q,a) + \veps'' 2^{\ifin+1}] & \mbox{by Eqns.~\eqref{eq:farpts},\eqref{eq:nearpts}}\\
&\leq \frac{1+\veps'}{1-\veps'}\sum_{q\in Q}d(q,a) +2\veps'' n 2^{\ifin} \\
&\leq \frac{1+\veps'}{1-\veps'}\sum_{q\in Q}d(q,a) + \frac{2\veps''}{(\alpha/2 - 3)(1-\beta)}\sum_{q\in Q}d(q,a) & \mbox{by Lemma~\ref{lemma:bound_opt}}\\
& =\Big(\frac{1+\veps'}{1-\veps'} + \frac{2\veps''}{(\alpha/2 - 3)(1-\beta)}\Big)\sum_{q\in Q}d(q,a).
\end{align*}
Setting $\veps'':= \frac{(\alpha/2 - 3)(1-\beta) \veps}{2}$ and
$\veps':= \frac{\veps}{\veps+2}$, we get that $\sum_{q\in Q} d(q,x) \leq (1+\veps) \sum_{q\in Q} d(q,a).$
\end{proof}

\subsection{Runtime Analysis}
In the first part, the compacted version of $T|_Q$ is constructed (in line 1)
in time $O(n\log n)$ using Lemma~\ref{lem:TQ}.
At each iteration $i$ we locate $\hat y$ (in line 9),
which becomes $y_\ifin$.
The runtime of this step is proportional to the number of candidates
in the $c$-list multiplied by the size $R_i$.
The number of candidates is $|L_{y,i,c}| \leq 2^{O(\ddim)}$.
The size of $R_{i}$ is at most the number of points in $Y_{i}$ which are at most $c'2^i$ away from $y_i$.
We conclude that $|R_i| \leq 2^{\ddim\log({c'2^{i+1}}/{2^{i}})} \le 2^{O(\ddim)}.$

Computing $R_i$ from $R_{i+1}$ takes $2^{O(\ddim)}$ time per member of $R_{i+1}$, for a total of $2^{O(\ddim)}$ per iteration $i$. Thus, the total time spent on finding $y_\ifin$ is $O(n\log n + 2^{O(\ddim)}\log \Delta).$
For the descendants search used in the refinement to $(1+\veps)$--approximation,
we can bound the number of representatives as follows.
\begin{lemma}\label{lem:1med_rep}
$|R_{\veps''}| \leq \veps^{-O(\ddim)}.$
\end{lemma}

\begin{proof}
If all of the representatives are in $Y_{\ifin-\log (1/\veps'')}$ then the size of $R_{\veps''}$ is at most the number of points in $Y_{\ifin-\log (1/\veps'')}$ which are at most $\frac{3c}{2\veps'}2^\ifin$ away from $y_\ifin$. The number of such points is bounded above by
\begin{align*}
2^{O(\ddim\log({\frac{3c}{2\veps'}2^\ifin}/{2^{\ifin-\log (1/\veps''}}))}
 = 2^{O(\ddim\log\frac{3c}{2\veps'\veps''})}
 \le \veps^{-O(\ddim)}.
\end{align*}
However, the representatives do not all have to be in $Y_{\ifin-\log (1/\veps'')}$. To overcome this, we charge each representative in $R_{\veps''}$ to a different point in $Y_{\ifin-\log (1/\veps'')}$. This mapping is done by assigning to each point in $R_{\veps''}$ its ancestor in the un-compacted $T|_Q$ which is in $Y_{\ifin-\log (1/\veps'')}$.
Notice that no two points in $R_{\veps''}$ can be assigned to the same point in $Y_{\ifin-\log (1/\veps'')}$, as otherwise there would be another node in the compacted $T|_Q$ which is an ancestor of those two points and in $Y_k$ for $k\leq{\ifin-\log (1/\veps'')}$, which contradicts the method in which the representatives were picked.
\end{proof}

The total cost of the descendants search is the number of representatives multiplied by the number of center candidates. From Lemma~\ref{lemma:descendant_search} we know that the number of candidates is at most $\veps''^{-O(\ddim)} \le \veps^{-O(\ddim)}$, and therefore, the runtime of this refinement stage is bounded by $\veps^{-O(\ddim)}$.

Overall, the runtime of computing a $(1+\veps)$--approximation for the 1-median is $O(n\log n) + 2^{O(\ddim)}\log \Delta(M) + \veps^{-O(\ddim)}$,
and this completes the proof of Theorem \ref{thm:eff1med}.

\section{Algorithm for $1$-Center}\label{app:1center}

It is helpful to see the solution for the $1$-center problem prior to seeing the solution for the $p$-center problem, as many of the ideas are similar, and the exposition with only one center is simpler. Therefore, we first prove the following.

\begin{theorem}
There is an algorithm that preprocesses a finite metric $M$
in time $2^{O(\ddim)}m\log \Delta \log\log \Delta(M)$ using $2^{O(\ddim)}m$ memory words,
where $m=|M|$, $\ddim=\ddim(M)$ and $\Delta=\Delta(M)$,
so that subsequent $1$-center queries on a set $Q\subseteq M$,
can be answered with approximation factor $1+\veps$,
for any desired $0<\veps\le1/2$,
in time $O(n\log n+ \log \log \log \Delta + \veps^{-O(\ddim)})$ ,
where $n=|Q|$.
\end{theorem}

The preprocessing algorithm simply builds the net hierarchy for the metric $M$, and prepares it for weighted level ancestor queries
(see Section \ref{sec:DS}). For the query, we first recall a trivial algorithm that provides a $2$--approximation
for the $1$-center problem on a query set $Q\subset M$,
and then refine it to provide a $(1+\veps)$--approximation.

Let $a\in M$ be an optimal center,
and denote its value by $\OPT:=\max_{q\in Q} d(a,q)$.
Notice that every point $y\in Q$ gives a $2$--approximation,
because its objective value is
\[
 \ALG_0
 :=\max_{q\in Q} d(y,q)
 \leq \max_{q\in Q} \aset{d(y,a) + d(a,y)}
 \leq 2\cdot \OPT.
\]
We thus pick any point $y\in Q$ as our first approximation,
and proceed to the refinement stage.

\paragraph{Refinement to $(1+\veps)$--approximation.}
Let $i$ be an integer such that  $2^{i-1}< \ALG_0 \leq 2^i$, and notice that $OPT\leq ALG_0\leq 2^i$. We begin by locating the ancestor $y_i\in Y_i$ of $y$ in $T$. This can be done using a weighted level ancestor query~\cite{FM96,KL07}. We next show that $a$ is fairly close to $y_i$.

\begin{lemma}
Let $a_{i-1}\in Y_{i-1}$ be an ancestor of $a$ in $T$. Then $a_{i-1} \in L_{y_i,i,6}$.
\end{lemma}
\begin{proof}
For every point $q\in Q$,
\[
  d(a_{i-1}, y_i)\leq d(a_{i-1},a) +d(a,q) + d(q,y)+d(y,y_i)
  \leq 2^{i} + \OPT+ 2\cdot\OPT+ 2^{i+1}
  \leq 6 \cdot 2^i.
\]
\end{proof}

This lemma implies that the optimal center $a$ is a descendant in $T$ of some point in $L_{y_i,i,6}$. Executing a descendants search from all the points in $L_{y_i,i,6}$ by using Lemma~\ref{lemma:descendant_search} with refinement constant $\veps$ will guarantee that we traverse a point $a_\veps$ such that $d(a,a_\veps)\leq \veps 2^i$.
Denote the set of the points seen in such a descendants search by $D$. Unfortunately, this process computes (separately) the cost of each candidate traversed by taking the maximum distances from all of $Q$ to that candidate,
which would take time $\veps^{-O(\ddim)}n$.
We can speed up this process by using (a few) representatives of $Q$,
as is explained next.

\paragraph{Speeding up the descendants search.}
We wish to find a bounded-size set of representatives for the points in $Q$,
such that the distortion caused by considering them (instead of $Q$) is small.
To this end, consider the set of representatives obtained as follows. Each point $q\in Q$ is mapped to its ancestor in the compacted $T|_Q$ which is in $Y_k$ for the largest $k\leq {i-\log({1}/\veps')}$, for some refinement constant $\veps'=\Theta(\veps)$ to be determined later. Call this set of representatives $R_{\veps'}$. Notice that $R_{\veps'}$ is a subset of the compacted $T|_Q$ and thus the process of this mapping can be done efficiently by scanning the compacted $T|_Q$ in linear time. Now, for each center candidate $x\in D$ we compute $\max_{r\in R_{\veps'}}d(r,x)$, and return the candidate $\hat x$ that minimizes this cost.

The next lemma shows that this algorithm achieves $(1+\veps)$--approximation.
\begin{lemma}
$\cntr(Q,\aset{\hat x}) = \max_{q\in Q}d(\hat x,q) \leq (1+\veps)\OPT$.
\end{lemma}
\begin{proof}
Every $q\in Q$ has a representative in $R_{\veps'}$,
for which we can apply Lemma~\ref{lemma:ancestor_descendant_distance_T}
and the triangle inequality, and thus
\[
  \max_{q\in Q}d(\hat x,q)
  < \max_{r\in R_{\veps'}} d(\hat x,r) + \veps' 2^{i+1}.
\]
Recall that one of the center candidates is some
$a_{\veps'}\in Y_{i-\log(1/\veps')}$ that is an ancestor of $a$ in $T$.
Therefore, the returned center $\hat x$ satisfies
\[
  \max_{r\in R_{\veps'}}d(\hat x,r)
  \leq \max_{r\in R_{\veps'}}d(a_{\veps'},r).
\]
Let $r^*\in R_{\veps'}$ be a maximizer for the righthand side,
and let $q^*\in Q$ be such that $r^*$ is a representative of $q^*$.
Using the triangle inequality and Lemma~\ref{lemma:ancestor_descendant_distance_T} again,
\[
  d(a_{\veps'},r^{*})
  \leq d(a,a_{\veps'}) + d(a,q^{*}) + d(q^{*},r^{*})
  \leq \OPT + 2\cdot \veps' 2^{i+1}.
\]
Recalling from earlier that $2^i < 2\ALG_0 \le 4\OPT$,
we finally combine the inequalities above and conclude that
$ \max_{q\in Q}d(\hat x,q)
  \leq \OPT + 3\cdot \veps 2^{i+1}
  \leq (1+24\veps)\OPT$.
To complete the proof, set $\veps'$ to be a power of $2$ in the range
$[\tfrac{\veps}{48},\tfrac{\veps}{24}]$.
\end{proof}

\paragraph{Runtime.}
The running time of the above query process is as follows. Locating $y_i$ using a weighted level ancestor query takes $O(\log \log \log \Delta)$ time as there are only $\log \Delta$ possible nets. After constructing $T|_Q$ in $O(n\log n)$ time, the mapping of every $q\in Q$ to its representative takes, altogether, $O(n)$ time.
The descendants search for each of the $O(2^{\ddim})$ points in $L_{y_i,i,6}$ takes time $\veps'^{-O(\ddim)} \leq \veps^{-O(\ddim)}$ time, which also bounds the number of candidates.
The number of representatives for $Q$ is bounded by the following lemma.

\begin{lemma}
$|R_{\veps'}| \leq \veps^{-O(\ddim)}.$
\end{lemma}
\begin{proof}
If all of the representatives are in $Y_{i-\log \frac{1}{\veps'}}$ then the size of $R_{\veps'}$ is at most the number of points in $Y_{i-\log \frac{1}{\veps'}}$ which are at most $2^i$ away from $y_i$. The number of such points is bounded above by
\begin{align*}
2^{\ddim\cdot\log({2^i}/{2^{i-\log(1/\veps')}})}
 = (1/\veps')^{O(\ddim)}
 = \veps^{-O(\ddim)}.
\end{align*}
However, the representatives do not all have to be in $Y_{\ifin-\log (1/\veps')}$. To overcome this, we charge each representative in $R_{\veps'}$ to a different point in $Y_{\ifin-\log (1/\veps')}$. This mapping is done by assigning to each point in $R_{\veps'}$ its ancestor in the un-compacted $T|_Q$ which is in $Y_{\ifin-\log (1/\veps')}$.
Notice that no two points in $R_{\veps'}$ can be assigned to the same point in $Y_{\ifin-\log (1/\veps')}$, as otherwise there would be another node in the compacted $T|_Q$ which is an ancestor of those two points and in $Y_k$ for $k\leq{\ifin-\log (1/\veps')}$, which contradicts the method in which the representatives were picked.
\end{proof}

It follows that the time it takes to evaluate the cost of all center candidates
is (altogether) $\veps^{-O(\ddim)}$,
and thus the algorithm's total runtime of is
$O(n\log n + \log \log \log \Delta + \veps^{-O(\ddim)})$.

\section{Algorithm for $p$-center}\label{sec:p_center}
\begin{theorem}\label{thm:p_center}
There is an algorithm that preprocesses a finite metric $M$
in time $2^{O(\ddim)}m\log \Delta\log\log \Delta$ using $2^{O(\ddim)}m$ memory words,
where $m=|M|$, $\ddim=\ddim(M)$ and $\Delta=\Delta(M)$,
so that subsequent $p$-center queries on a set $Q\subseteq M$,
can be answered with approximation factor $1+\veps$,
for any desired $0<\veps\le1/2$,
in time $O(n \log n  + p\log \log \log \Delta + p^{p+1} \veps^{-O(p\cdot \ddim)})$ ,
where $n=|Q|$.
\end{theorem}

The preprocessing algorithm simply builds the net hierarchy for the metric $M$, and prepares it for weighted level ancestor queries
(see Section \ref{sec:DS}). For the query, we first use the algorithm of Gonzalez from~\cite{Gonzalez85} on $Q$, which obtains a $2$-approximation for the $p$-center in $O(p\cdot n)$ time. In other words, the algorithm locates a set $B\subset Q$ of size $p$ such that if we denote its objective value as $\ALG_0 :=\max_{q\in Q} d(q,B)$, and if $A\subset M$ is an optimal $p$-center set with value $\OPT:=\max_{q\in Q} d(q,A)$, then $\ALG_0 \leq 2\cdot \OPT$.

\subsection{Refinement to $(1+\veps)$--approximation.}
Let $i$ be an integer such that  $2^{i-1}<\ALG_0 \leq 2^i$. For each $b\in B$ locate the ancestor $b_i\in Y_i$ of $b$ in $T$. This can be done using a weighted level ancestor query~\cite{FM96,KL07}. For a center $a\in \OPT$, let $a_{i-1}\in Y_{i-1}$ be an ancestor of $a$ in $T$.

\begin{lemma}
For every $a\in \OPT$ there exists a point $b\in B$ such that $a_{i-1} \in L_{b_i,i,6}$.
\end{lemma}
For every point $q\in Q$ which is assigned to $a$ in OPT, let $b$ be the center of the cluster of $q$ in $B$. Then $d(a_{i-1}, b_i)\leq d(a_i,a) +d(a,q) + d(q,b)+d(b,b_i)\leq 2^{i} + \OPT+ 2\cdot \OPT+ 2^{i+1} \leq 6 \cdot 2^i$. \qed

This implies that a center $a\in A$ is a descendant in $T$ of some point in $\bigcup_{b\in B} L_{b_i,i,6}$. Performing a descendants search from each of the points in $L_{b_i,i,6}$ by using Lemma~\ref{lemma:descendant_search} for some refinement constant $\veps'=\theta(\veps)$ to be determined later, will guarantee that for each $a\in\OPT$ we traverse a point $\hat{a}$ such that $d(a,\hat{a})\leq \veps' 2^i$. Denote the union of the points seen in such a descendants search by $D$. Unfortunately, this process computes (separately) the cost of each subset of size $p$ of candidates traversed by taking the maximum distances from all of $Q$ to that subset, which would take time $\veps^{-O(\ddim)}np$.
We can speed up this process by using (a few) representatives of $Q$.

\paragraph{Speeding up the descendants search.}

We wish to find a bounded-size set of representatives for the points in $Q$,
such that the distortion caused by considering them (instead of $Q$) is small.
To this end, consider the set of representatives obtained as follows. Each point $q\in Q$ is mapped to its ancestor in the compacted $T|_Q$ which is in $Y_k$ for the largest $k\leq {i-\log({1}/\veps')}$, for some refinement constant $\veps'=\Theta(\veps)$ to be determined later. Call this set of representatives $R_{\veps'}$. Notice that $R_{\veps'}$ is a subset of the compacted $T|_Q$ and thus the process of this mapping can be done efficiently by scanning the compacted $T|_Q$ in linear time. Now, for each set of $p$ center candidates $X\subset D$ we compute $\max_{r\in R_{\veps'}}d(r,X)$, and take the set of candidates $\hat{X}$ that minimizes this cost.

The next lemma shows that this algorithm achieves $(1+\veps)$--approximation.
\begin{lemma}
$\cntr(Q,\aset{\hat X}) = \max_{q\in Q}d(\hat X,q) \leq (1+\veps)\OPT$.
\end{lemma}

\begin{proof}
Every $q\in Q$ has a representative in $R_{\veps'}$,
for which we can apply Lemma~\ref{lemma:ancestor_descendant_distance_T}
and the triangle inequality, and thus
\[
  \max_{q\in Q}d(\hat X,q)
  < \max_{r\in R_{\veps'}} d(\hat X,r) + \veps' 2^{i+1}.
\]
Recall that one of the sets of center candidates is some
$A_{\veps}\subset Y_{i-\log(1/\veps')}$ that is the set of ancestors of every $a\in A$ in $T$, where $A$ is an optimal solution.
Therefore, the returned center set $\hat X$ satisfies
\[
  \max_{r\in R_{\veps'}}d(\hat X,r)
  \leq \max_{r\in R_{\veps'}}d(A_{\veps'},r).
\]
Let $r^*\in R_{\veps'}$ be a maximizer for the righthand side,
and let $q^*\in Q$ be such that $r^*$ is a representative of $q^*$. Let $a\in A$ be the center in $A$ which is closest to $q^*$, and let $a_{\veps'}$ be the ancestor of $a$ in $A_{\veps'}$.
Using the triangle inequality and Lemma~\ref{lemma:ancestor_descendant_distance_T} again,
\[
  d(A_{\veps'},r^{*})
  \leq d(a_{\veps'},a) + d(a,q) + d(q,r^{*})
  \leq \OPT + 2\cdot \veps' 2^{i+1}.
\]
Recalling from earlier that $2^i < 2\ALG_0 \le 4\OPT$,
we finally combine the inequalities above and conclude that
$ \max_{q\in Q}d(\hat X,q)
  \leq \OPT + 3\cdot \veps 2^{i+1}
  \leq (1+24\veps)\OPT$.
To complete the proof, set $\veps'$ to be a power of $2$ in the range
$[\tfrac{\veps}{48},\tfrac{\veps}{24}]$.
\end{proof}

\subsection{Runtime}
The running time of the above process is as follows. Locating $b_i$ for all $b\in B$ using a weighted level ancestor queries takes $O(p\log \log \log \Delta)$ as there are only $\log \Delta$ possible nets. After constructing $T|_Q$ in $O(n\log n)$ time, the mapping of each $q\in Q$ to its representative takes another $O(n)$ time. The descendants search from all of the $O(2^{\ddim})$ points in $L_{b_i,i,6}$ takes $O(p\veps'^{-O(\ddim)}) = O(p\veps^{-O(\ddim)})$, which also bounds the number of candidates. The number of representatives can be bounded by the following lemma.

\begin{lemma}
$|R_{\veps'}| \leq p\veps^{-O(\ddim)}.$
\end{lemma}
\begin{proof}
If all of the representatives are in $Y_{\ifin-\log (1/\veps')}$ then the size of $R_{\veps'}$ is at most the number of points in $Y_{i-\log ({1}/{\veps'})}$ which are at most $2^i$ away from each of the $p$ points in $B$. The number of such points is bounded above by
\begin{align*}
p\veps'^{-O(\ddim)}
 \leq p\veps^{-O(\ddim)}.
\end{align*}
However, the representatives do not all have to be in $Y_{\ifin-\log (1/\veps')}$. To overcome this, we charge each representative in $R_{\veps'}$ to a different point in $Y_{\ifin-\log (1/\veps')}$. This mapping is done by assigning to each point in $R_{\veps'}$ its ancestor in the un-compacted $T|_Q$ which is in $Y_{\ifin-\log (1/\veps')}$.
Notice that no two points in $R_{\veps'}$ can be assigned to the same point in $Y_{\ifin-\log (1/\veps')}$, as otherwise there would be another node in the compacted $T|_Q$ which is an ancestor of those two points and in $Y_k$ for $k\leq{\ifin-\log (1/\veps')}$, which contradicts the method in which the representatives were picked.
\end{proof}

Thus, the time it takes to test each of the ${p\veps^{-O(\ddim)} \choose p}$ candidates is at most $O(p\veps^{-O(\ddim)})$, and the total runtime of the algorithm is $O(n\log n + p\log \log \log \Delta + p^{p+1}\veps^{-p\cdot \ddim})$. Notice that the runtime of the algorithm of Gonzalez is $O(np)$, which is always bounded from above by $O(n\log n+p^{p-1})$, and can thus be absorbed by the other terms.

\section{Algorithm for $p$-median}
\label{app:pmedian}

\begin{theorem}\label{thm:p_median}
There is an algorithm that preprocesses a finite metric $M$
in time $2^{O(\ddim)}m\log \Delta\log\log \Delta$ using $2^{O(\ddim)}m$ memory words,
where $m=|M|$, $\ddim=\ddim(M)$, and $\Delta=\Delta(M)$,
so that subsequent $p$-median queries on a set $Q\subseteq M$ of size $n$,
can be answered within approximation factor $1+\veps$,
for any desired $\veps\in(0,\tfrac12]$,
in time $O(n\log n) + \veps^{-O(\ddim)} (p\cdot \log n)^{O(1)} \cdot \log \log \log \Delta + {\veps ^{-O(p\cdot \ddim)}} (p\cdot \log n)^{O(p)}$.
\end{theorem}

To a large extent, we follow an algorithm of Har-Peled and Mazumdar~\cite{HM04}
for approximating $p$-median clustering in Euclidean space.
Their algorithm runs in time roughly
$O(n + \operatorname{exp}(\veps^{-d}) (p\cdot \log n)^{O(1)})$,
where $d$ is the dimension in Euclidean space
(in a scenario without preprocessing).
In order to give a flavor of our preprocessing model,
we focus on the case of small $p$ and employ an abridged version of their
algorithm, with runtime that grows exponentially with $p$.
We note that following their techniques more closely may possibly
reduce the runtime, like eliminating the exponential dependence on $p$.

\begin{proof}[Proof (Sketch)]
At a high level,
the algorithm of Har-Peled and Mazumdar~\cite{HM04} works as follows.
First, construct a set $A$ of $\hat p:=p \cdot \log ^{O(1)} n$ centers
that provides a constant factor approximation of the $p$-median
(formally, it is a bicriteria approximation, since $\hat p>p$).
Next, construct a core-set $S$ by building an exponential grid
(as defined below) around each of the centers in $A$,
and mapping each point in $Q$ to its (approximate) closest grid point,
using near neighbor search.
The size of the core-set is roughly $|S|\le \veps^{-d} |A|\log n$
(but of course these points have weights that add up to $n$).
This means that every solution to the $p$-median problem on $S$ is
a good approximation for the $p$-median problem on $Q$.
Finally, construct another set of exponential grids around each of the points
in $S$ to obtain a centroid set $D$,
i.e., set of potential centers in the ambient (Euclidean) space,
of size roughly $|D| \le \veps^{-d}|S|^{O(1)} \le \veps^{-2d} (p\cdot \log n)^{O(1)}$.
Finally, use a variant of the dynamic programming algorithm of
Kolliopoulos and Rao~\cite{KR07} to quickly compute a near-optimal
$p$-median of $S$ among the potential centers $D$.

This algorithm of Har-Peled and Mazumdar~\cite{HM04} carries over to our scenario (possibly using some different black-box data structures, for example to solve nearest neighbor search), except for the following two main ingredients.
The first is the construction of the exponential grid, which strongly relies on being in Euclidean space, and the ability to define points in ambient space, which we do not enjoy in doubling dimension metrics.
The second is the dynamic programming solution of Kolliopoulos and Rao~\cite{KR07}, which also exploits the Euclidean space structure.
We solve the exponential grid using $T$, as shown below, and skip the use of dynamic programming by performing a brute-force search over all size $p$ subsets (of the centroid set).
It is plausible that our runtime can be improved by adapting the solution of
Kolliopoulos and Rao~\cite{KR07} to work in our case as well,
and we leave this for future work.

\paragraph{Exponential grid.}
The exponential grid of Har-Peled and Mazumdar~\cite{HM04} around a point $r$ with length parameter $R>0$ roughly works as follows.
They build $O(\log n)$ axis-parallel squares,
where the $j^{th}$ square has side length of $R2^j$ and
is partitioned into sub-squares (i.e., a grid) of side length $O(\veps R2^j/d)$,
the idea being that areas closer to the point $r$ have smaller cell size, while areas further away have larger cell size. This construction does not carry over to doubling dimension metrics as we cannot define grid points in ambient space. However, we make use of $T$ to provide a set with similar properties.

We provide a sketch of the idea in order to ease presentation, but point out that some of our constants can to be refined.
Given $r$, we use a weighted level ancestor query~\cite{FM96,KL07} to locate its ancestor $r_{\log R}\in Y_{\log R}$ in $T$. A descendants search, using Lemma~\ref{lemma:descendant_search} starting from $r_{\log R}$ with refinement constant $\epsilon$ will give us a good resolution for points that are roughly at most distance $R$ away from $r$. Let $r_{j+\log R}\in Y_{j+\log R}$ be an ancestor of $r$ in $T$, for $0\leq j \leq O(\log n)$. We perform a descendants search using Lemma~\ref{lemma:descendant_search} starting from each such $r_{j+\log R}$, with refinement constant $\veps$. Notice that for a specific $j$, for points within distance $2^{j+\log R} = 2^j R$ from $r$,
the descendant search starting from $r_{j+\log R}$ reaches a set of points which are in $Y_{j+\log R - \log (1/\veps)}$ which is similar to the resolution obtained from the exponential grid in Euclidean space.

The union of all of the points seen during all of the descendants searches on all $O(\log n)$ levels provides a set of size $\veps^{-O(\ddim)}\log n$,
which gives us (i.e., in doubling dimension metrics)
the same properties as the exponential grid does in a Euclidean space.
Thus we obtain a core-set $S$ and centroid set $D$ both of size at most
$\veps^{-O(\ddim)}(p\cdot \log n)^{O(1)}$.
Finally, perform an exhaustive search through all subsets of size $p$ of the centroid set $D$ and compute the cost of each such set, which takes total time
$\binom{|D|}{p}\cdot O(|S|p) \leq \veps ^{-O(p\cdot \ddim)}(p\cdot \log n)^{O(p)}$.
Notice that a weighted level ancestor query is performed for each of the points in the core-set $S$, which increases the runtime by
$\veps^{-O(\ddim)}(p\cdot \log n)^{O(1)} \cdot \log \log \log \Delta$.
\end{proof}

\
{\small
\bibliographystyle{alphainit}
\bibliography{robi}
}

\end{document}